\documentclass{elsarticle}
\usepackage{amsmath,amsthm}
\usepackage{amsfonts}
\usepackage{amssymb}

\usepackage{url}
\usepackage{hyperref}

\theoremstyle{plain}
\newtheorem{theorem}{Theorem}
\theoremstyle{definition}
\newtheorem{lemma}[theorem]{Lemma}
\newtheorem{definition}[theorem]{Definition}
\newtheorem{remark}[theorem]{Remark}

\newtheorem{corollary}[theorem]{Corollary}
\newtheorem{example}[theorem]{Example}

\newtheorem{proposition}[theorem]{Proposition}

\journal{Journal of Symbolic Computation}

\usepackage[T1]{fontenc}
\usepackage[utf8]{inputenc}
\usepackage{thmtools}
\usepackage{tikz-cd}
\usepackage{hhline}

\newenvironment{algorithm}[1][\relax]{%
  \medskip\par\noindent\underline{Algorithm}%
  \ifx#1\relax\relax\else:~#1\fi}{\par\smallskip}

\usepackage{multicol,multirow}
\usepackage{enumitem}

\usepackage[bordercolor=TodoColour, 
 backgroundcolor=TodoColour,  
 linecolor=TodoColour, 
 textsize=todosize]{todonotes}

\newlist{steps}{enumerate}{1}
\setlist[steps]{%
  label={\textbf{Step~\arabic*.}},%
  ref={step~\arabic*},%
  leftmargin=1.4142cm}
\setlist[enumerate,1]{label={(\alph*)}}

\usepackage{xcolor}
\definecolor{TodoColour}{HTML}{CE9F6F} 
\definecolor{CSTodoColour}{HTML}{78779B}
\definecolor{MyTodoColour}{HTML}{A5BA93}
\definecolor{LinkColour}{HTML}{044F67}
\definecolor{CiteColour}{HTML}{044F67}
\definecolor{URLColour}{HTML}{044F67}
\definecolor{EmphColour}{HTML}{86ABA5} 
\definecolor{SuperEmphColour}{HTML}{89729E}
\definecolor{darkblue}{rgb}{0.0, 0.0, 0.55}

\usepackage{hyperref}

\hypersetup{
  final,
  colorlinks, 
  urlcolor=black,
  linkcolor=black,
  citecolor=black}

\DeclareMathOperator{\diag}{diag}

\DeclareMathOperator{\denom}{denom}

\DeclareMathOperator{\hyp}{hyp}
\DeclareMathOperator{\trunc}{Trunc}
\DeclareMathOperator{\slc}{slc}

\renewcommand{\leq}{\leqslant}
\renewcommand{\geq}{\geqslant}
\renewcommand{\kappa}{\varkappa}
\renewcommand{\rho}{\varrho}

\newcommand{\qqtext}[1]{\qquad\text{#1}\qquad}
\newcommand{\qqtextq}[1]{\qquad\text{#1}\quad}

\newcommand{\qqwhere}{\qqtextq{where}}

\renewcommand{\)}{{)\!)}}

\newcommand{\Mat}[3]{#1^{#2\times#3}}
\newcommand{\MatGr}[2]{\operatorname{GL}_{#2}(#1)}

\newcommand{\ID}[1][\relax]{\mathbf{1}\ifx#1\relax\relax\else_{#1}\fi}
\def\ZEROAUX#1,#2;{_{#1\times#2}}
\newcommand{\ZERO}[1][\relax]{\mathbf{0}\ifx#1\relax\relax\else\ZEROAUX#1;\fi}

\newcommand{\RHS}{right-hand side}
\newcommand{\LHS}{left-hand side}
\newcommand{\lhs}{l.h.sol }

\newcommand{\Kbar}{\overline{K}}

\newcommand{\F}{\mathbb{F}}

\newcommand{\N}{\mathbb{N}}
\newcommand{\Z}{\mathbb{Z}}

\usepackage{latexsym,amscd,url}

\newcommand{\Ct}{{C} }
\newcommand{\De}{\delta}
\newcommand{\T}{\tau}
\newcommand{\GL}{\mathop{\rm GL}}

\newcommand{\Or}{\mathop{\rm O}}
\newcommand{\ord}{\mathop{v}} 

         {}
\newcounter{ex}

\def\wt{\widetilde}
\def\wh{\widehat}

\renewcommand{\)}{{)\!)}}

\usepackage{latexsym,amscd,url}

\def\wt{\widetilde}
\def\wh{\widehat}

\definecolor{benimidori}{HTML}{78779B}

\begin{document}

\begin{frontmatter}
\title{Hypergeometric Solutions of Linear Difference Systems \\
{\large In Memory of Marko Petkov\v{s}ek and Manuel Bronstein}}

\author{Moulay Barkatou}
\ead{moulay.barkatou@unilim.fr}
\affiliation{organization={Universit\'e de Limoges, XLIM},addressline={123, Av.  A. Thomas},city={87060 Limoges cedex},
country={France}}

\author{Mark van Hoeij\fnref{fn1}}
\ead{hoeij@math.fsu.edu}
\affiliation{organization={Florida State University},
city={Tallahassee, FL},postcode={FL 32306},
country={USA}}
\fntext[fn1]{Supported by NSF grant 2007959.}

\author{Johannes Middeke}
\ead{jmiddeke@risc.jku.at}
\affiliation{organization={RISC, Johannes Kepler University},
city={4040 Linz},
country={Austria}}

\author{Yi Zhou}
\ead{yzhou@math.fsu.edu}
\affiliation{organization={Florida State University},
city={Tallahassee, FL},postcode={FL 32306},
country={USA}}

\maketitle
 
\begin{abstract}
  We extend Petkov\v{s}ek's algorithm for computing hypergeometric solutions of scalar difference equations
  to the case of difference systems $\tau(Y) = M Y$, with $M \in \GL_n(\Ct(x))$, where $\tau$ is the shift operator.
 Hypergeometric solutions are solutions of the form $\gamma P$ where $P \in \Ct(x)^n$ and $\gamma$ is a hypergeometric term over $\Ct(x)$,
 i.e. ${\tau(\gamma)}/{\gamma}  \in \Ct(x)$.
 Our contributions concern efficient computation of a set of candidates for ${\tau(\gamma)}/{\gamma}$ which we write as $\lambda = c\frac{A}{B}$ with monic $A, B \in \Ct[x]$, $c \in \Ct^*$.
Factors of the denominators of $M^{-1}$ and $M$ give candidates for $A$ and $B$, while another algorithm is needed for $c$.
 We use super-reduction algorithm
 to compute candidates for $c$, as well as other ingredients to reduce the list of candidates for $A/B$. 
 To further reduce  the number of candidates $A/B$, we bound the {\em type} of $A/B$ by bounding {\em local types}.   
Our algorithm has been implemented in Maple and experiments show that our implementation can handle systems of high dimension, which is useful for factoring operators. 
\end{abstract}

\end{frontmatter}

\section{Introduction}\label{sec:hypsolintro}

Let $C$ be a subfield of $\mathbb{C}$. Let $\tau$ be the shift operator, which sends
a function $y(x)$ to $y(x+1)$.
A nonzero element $y$ in a universal extension  \cite[Chapter 6]{van2003galois} 
of $C(x)$ is called \emph{hypergeometric} if $\tau(y)/y\in C(x)$, in other words, $y$ is a solution of a first order operator
over $\Ct(x)$.
If $y_1, y_2$ are hypergeometric then $y_1 y_2$ is hypergeometric as well.
For $\lambda \in C(x)$, denote by $\hyp(\lambda)$ a nonzero solution of $\tau-\lambda$. 
The notation is defined up to a nonzero constant.

The first algorithm to compute hypergeometric solutions of an operator was given by
Pet\-kov\-\v{s}ek  \cite{Petkovsek1992}.  One writes hypergeometric solutions of $L = a_n \tau^n + \cdots + a_0 \tau^0 \in C[x][\tau]$ over $C(x)$ in this format:
$$
  y = \hyp(\lambda)P
$$
where $\lambda = c \frac{A}{B}$ with
$A,B,P \in C[x]$,  $A$ and $B$ are coprime and monic, and $c \in C^* = C \setminus \{0\}$.
Such $y$ can be rewritten as $\hyp(\lambda')$ where $\lambda' = \lambda \frac{\tau(P)}P$.
Allowing a polynomial factor $P$ in candidate solutions $y = \hyp(\lambda)P$
makes it easier to restrict $\lambda$ to a computable set. Petkov\v{s}ek's algorithm
works as follows:
\begin{steps}[leftmargin=0.75in]
\item[Step P1] Petkov\v{s}ek proves that it suffices to consider $A$, $B$ where $A \mid a_0$ and  $B \mid \tau^{1-n}(a_n)$.
  This leaves a finite set of candidates for $A/B$.
\item[Step P2]\label{step:2} Compute candidates for $c$.
\item[Step P3]\label{step:3} For each candidate $\lambda = c A/B$: \\
  Construct an operator $L_\lambda$ whose solutions are the
  solutions of $L$ divided by $\hyp(\lambda)$.
  For all polynomial solutions $P$ of $L_\lambda$: \ join  $\hyp(\lambda) P$ to the output.
\end{steps}

Consider a  $n$-dimensional system over $C(x)$:
\begin{equation}
  \tag{\textsc{sys}}
  \label{eq:sys}
  \tau(Y) = M Y,
  \qqwhere
  M \in \MatGr{C(x)}n.
\end{equation}
The goal of this paper is to design efficient algorithms for finding its hypergeometric solutions. 
A hypergeometric solution has the form
$\hyp(\lambda)P$
where $P$ is an $n$-dimensional column vector over $C(x)$.
If $\tilde{P} = r P$ and $r, \lambda \in \Ct(x)$
then $\hyp(\lambda)\tilde{P} = \hyp(\lambda \frac{\tau(r)}r) P$ so
we may assume without loss of generality that $P$ is in $C[x]^n$ and is a {\em primitive vector}
(the {\em content}, the gcd of entries, is 1).  

Before his sudden passing, Bronstein observed that Petkov\v{s}ek's strategy works for systems as well \cite{bronsteinprivate}
by replacing the bounds $a_0$ and $\tau^{1-n}(a_n)$ for $A$ and $B$ with the denominators of $M^{-1}$ and $M$.
Thus, the {\em Bronstein-Petkov\v{s}ek} strategy for hypergeometric solutions is as follows:
\begin{steps}
\item[BP1]\label{step:1}  Construct a set of candidates for $A/B$: 
$$\mathcal{S}:=\{\frac{A}{B}: A,B\in C[x]\text{ are monic}, \ A\mid \denom(M^{-1}), \ B\mid \denom(M)\}.$$
\item[BP2] For each candidate $A/B$ compute candidates for $c$.
\item[BP3] For each candidate $\lambda = c A/B$: compute all polynomial solutions $P \in C[x]^n$ of $\tau(P) = \lambda^{-1} M P$
and join  $\hyp(\lambda) P$ to the output.
\end{steps}
Remarkably, Bronstein's proof for Step BP1 (Theorem~\ref{thm:AB})
is actually easier than Petkov\v{s}ek's proof for Step P1, despite the fact that it is more general.
To obtain an algorithm, Bronstein still needed a way to find candidates for $c$, which we will give in Section~\ref{sec:hypsolgenexp}.

The BP-strategy also applies to other cases (e.g. $q$-difference systems, or the multi-basic case \cite{denomBound})
provided that one can compute candidates for $c$, and polynomial solutions.
For operators $L \in C(x)[\tau]$, the paper \cite{VANHOEIJ1999109} addressed the following issues in Petkov\v{s}ek's algorithm:
\begin{enumerate}
\item The number of candidates $A/B$ can be much larger than it needs to be.
\item As a side effect, the algorithm can produce duplicate solutions.
\end{enumerate}
For systems, the same issues arise in the BP-strategy. A goal in this paper is to address these issues.
Of course one might discard duplicate solutions, or take steps to prevent them, but that still leaves issue (a). 
Reducing the number of candidates as much as possible leads to a more efficient algorithm and eliminates issue (b) as a side effect.
A key idea is that rather than bounding $A/B$ using the denominators of $M^{-1}$ and $M$,
we bound the {\em type} of $A/B$ by bounding {\em local types}.

In \cite{AbramovPetkovsekRyabenko2015} another algorithm for computing hypergeometric solutions of difference systems is presented. It consists in reducing the input system to one or more scalar equations and then use Petkov\v{s}ek's algorithm. 
This is similar to the cyclic vector method, which limits the dimension for which the algorithm can be efficient
(see Section~\ref{SectionNew}).
Experiments show that our algorithm (implementation at~\cite{HypSolsSystems})
can handle systems of high dimension, which will be useful for the application in factoring operators discussed in Section~\ref{sec:bronsteindifference}.

\section{Hypergeometric Solutions}\label{sec:hyper}

\begin{definition}[Hypergeometric]\label{def:hg}
Let $\F$ be a difference field, such as $\F=C(x)$ or $C((x^{-1}))$. A non-zero element $\gamma$ in a universal extension of $\F$
 is called \emph{hypergeometric} if $\tau(\gamma)=f\gamma$ for some $f\in\F^*$. In this case let $\hyp(f)$ denote $\gamma$. The notation is unique up to a constant.
  When $\gamma$ is hypergeometric, call the column vector $\gamma R$ (where $R\in \F^n$) a \emph{hypergeometric vector}.
\end{definition}
We focus on computing hypergeometric solutions because this is the key algorithmic step needed in other algorithms,
such computing Liouvillian solutions or factoring (Section~\ref{SectionNew}).
It is not hard to verify, for $f, f_1, f_2 \in \F^*$: 
\begin{itemize}
    \item[(i)] $\hyp(f_1)\hyp(f_2)=\hyp(f_1f_2)$ up to a constant 
    \item[(ii)] $\hyp(\frac{\tau(f)}{f})=f$ up to a constant 
\end{itemize}

The formula
   $$ \hyp(f)=g\hyp\left(f\frac{g}{\tau(g)}\right)$$
shows that a hypergeometric vector can have different representations.
A hypergeometric vector over $C(x)$ can always be written as
$\gamma P$ where $\gamma$ is hypergeometric and $P\in C[x]^n$ is primitive. This is called the \emph{standard representation} of the hypergeometric vector.
It is unique up to a constant.

\section{Algorithm Version I}\label{sec:algv1}
In this section a basic version of the algorithm is presented, which follows the procedure in Section~\ref{sec:hypsolintro}. We first give the algorithm and then explain it in the follow-up sections.

\begin{algorithm}[Version I]~\\
Input: $\tau(Y)=MY$, where $M\in \GL_n(C(x))$ \\
Output: hypergeometric solutions of the system 
\begin{itemize}[leftmargin=1cm]
    \item[BP1]  
        Factor $d_1 := \denom(M)$ and $d_2 := \denom(M^{-1})$ in $C[x]$ to compute 
            $$\mathcal{S}:=\{\frac{A}{B}: A,B\in C[x]\text{ are monic}, A\mid d_2, B\mid d_1\}.$$
        In Section~\ref{sec:V1BP1} we show that $\mathcal{S}$ contains all $\frac{A}{B}$ that are needed.
    \item[BP2]    
        Generalized exponents and ``{slc}'' will be defined in Section~\ref{sec:hypsolgenexp}.
        Use the algorithm introduced in subsection \ref{sec:d} to compute 
        $$G:=\{\text{unramified generalized exponents of the system}\}.$$
        An {\em unramified generalized exponent} is in the form $cx^s(1+dx^{-1})$, where $(s,c,d)\in  \mathbb{Z}\times \Ct^*\times \Ct$.
        Let $\mathcal{H}:=$
        $$\{c\frac{A}{B}:cx^s(1+dx^{-1})\in G,\frac{A}{B}\in \mathcal{S}, \deg(A)-\deg(B)=s,d-\slc(\frac{A}{B})\in \mathbb{N}\}.$$
    \item[BP3]   
        Let $\mathrm{Sols}=\emptyset$. \\
        For each $c\frac{A}{B}\in \mathcal{H}$, solve the system $\tau(P)= c^{-1}\frac{B}{A}MP$ for (a basis of) polynomial solutions using the algorithm introduced in  \cite{AB98} (see also \cite{abramov1999}). Add $\hyp(c\frac{A}{B})P$ to $\mathrm{Sols}$ for any polynomial solution $P$. \\
        Return $\mathrm{Sols}$ as the output.
\end{itemize}  
\end{algorithm}

\subsection{Step BP1}\label{sec:V1BP1} 

\begin{theorem}\label{thm:AB}
  Suppose $Y=\hyp(c \frac{A}B) P$ is a hypergeometric solution 
  of equation~(\ref{eq:sys}) that is in the standard form, where $c \in C^*$ and $A,B \in C[x]$ are monic and coprime. Then
  $$A \mid \denom(M^{-1}),\quad B \mid \denom(M).$$
\end{theorem}
The theorem says that any hypergeometric solution $Y$ can be written as $Y = \hyp(c \frac{A}{B}) P$ for some primitive polynomial vector $P$, some constant $c$, and some $\frac{A}{B}\in\mathcal{S}$,
with $\mathcal{S}$ from Step BP1. 

The result holds not only for $C(x)$, but also for the field of fractions of a general difference ring if it is a UFD.
\begin{proof}
   Write now $M = d_1^{-1} W$ where $d_1 = \denom(M)$. Then, substitution of
this and $Y = \gamma P$ in standard form into equation~(\ref{eq:sys}) allows
us to rewrite the equation as
$$
  \frac{\tau(\gamma)}\gamma \tau(P) = \frac1{d_1} W P
  \qqtext{or}
  d_1 c A \tau(P) = B W P
$$
using $\tau(\gamma)/\gamma = cA/B$ and clearing denominators.
Now $B$ divides the \RHS\ of the equation. We assumed $A$ and $B$ are coprime, and so $B$ must divide the factor $d_1 \tau(P)$ in the \LHS.
But $P$ and hence $\tau(P)$ is primitive, and so $B \mid d_1$.

Write the inverse of $M$ as $M^{-1} = d_2^{-1} W'$ where $d_2 \in C[x]$
and $W' \in \Mat{C[x]}nn$. Then, multiplying equation~(\ref{eq:sys}) by $M^{-1}$ and substituting again
$Y = \gamma P$ in standard form, we obtain
$$
  \frac{1}{d_2} W' \tau(P)
  =
  \frac{B}{cA}  P
  \qqtext{or}
  cA W' \tau(P) = d_2 B P.
$$
As before, we can derive $A \mid d_2$.
\end{proof}

\begin{example}\label{example1}

Let $M = $
{\small
\begin{multline*}
\mbox{}\hspace{-8pt}\left(\begin{smallmatrix}
  0 & 0 & \tfrac{x+1}{x} & 0 & 0 & 0 \\
  0 & 0 & 0 & 0 & \tfrac{x+1}{x} & 0 \\
  0 & 0 & 0 & 0 & 0 & \tfrac{x+1}{x} \\
  \tfrac{(x+1)^2}{x(x+2)(x+3)} \hspace{-0.5cm} & 0 & \hspace{-0.5cm}-\tfrac{(x+1)(x^4+2x^3+x^2+x+4)}{x(x+2)^2(x+3)}\hspace{-0.5cm} & 0 & -\tfrac{x+1}{x(x+2)^2(x+3)} & 0 \\
  0 & \hspace{-0.2cm}\tfrac{(x+1)^2}{x(x+2)(x+3)}\hspace{-0.2cm} & \tfrac{(x+1)^2}{x(x+2)^2(x+3)} & 0 & 0 & -\tfrac{x+1}{x(x+2)^2(x+3)} \\
  0 & 0 & 0 &\hspace{-0.1cm} \tfrac{(x+1)^2}{x(x+2)(x+3)}\hspace{-0.1cm} & \hspace{-0.2cm}\tfrac{(x+1)^2}{x(x+2)^2(x+3)} \hspace{-0.2cm}& \tfrac{(x+1)(x^4+2x^3+x^2+x+4)}{x(x+2)^2(x+3)} \\
\end{smallmatrix}\right)
\end{multline*}
}
Then $\denom(M)=x(x+2)^2(x+3)$ and $\denom(M^{-1})=(x+1)^2(x+2)$ and so
$\mathcal{S}=\{x^{i_0}(x+1)^{i_1}(x+2)^{i_2}(x+3)^{i_3}:-1\leq i_0\leq 0, 0\leq i_1\leq 2,-2\leq i_2\leq 1, -1\leq i_3\leq 0\}$
whose cardinality is
$$2\cdot 3\cdot 4\cdot 2=48.$$
This number is larger than necessary; a provably complete search can be done with just 8 cases, see Example~\ref{example4} which addresses issues (a) and (b) stated in Section~\ref{sec:hypsolintro}.
\end{example}

\begin{remark}
Finding hypergeometric solutions of the system in Example~\ref{example1} is related to the factorization of the operator
$$(x+2)^2(x+3)\tau^4+\tau^3+(x^4+2x^3+x^2+x+4)\tau^2+(x+1)\tau+(x+1)(x+2).$$
It is explained in Section~\ref{sec:bronsteindifference} how to convert a factorization problem into solving a system for hypergeometric solutions. 
\end{remark}


\subsection{Step BP2: generalized exponents}\label{sec:hypsolgenexp}
The key of Step BP2 is the notion of \emph{generalized exponents} for systems.
For operators, a definition is given in \cite[Section 3.2]{YCHA}, for implementations see the references in Section~\ref{SectionNew}.
We will define generalized exponents of a system in a similar way. We start with some facts about the difference field $\mathbb{C}((x^{-1}))$. 

\subsubsection{The difference field $K$}
Denote $t=1/x$. Let $K=\mathbb{C}((t))$ and $K_r=\mathbb{C}((t^\frac{1}{r}))$ for $r=1,2,3,\ldots$. According to \cite{walkeralgcurves} the algebraic closure of $K$ is
$$\overline{K} = \bigcup_{r=1}^\infty K_r.$$
(This section uses $\mathbb{C}$ instead of $C$ because the above formula holds if the field of constants is algebraically closed.)
We will also denote $\overline{K}$ as $K_\infty$.  
The canonical $t$-adic valuation on $K$ extends to $K_r$ naturally; we denote it by $v:\Kbar\mapsto \mathbb{Q}\cup \{\infty\}$. For $v$ to be a valuation, it has to satisfy the following properties:
\begin{itemize}
    \item $v(a)=\infty$ if and only if $a=0$,
    \item $v(ab)=v(a)+v(b)$,
    \item $v(a+b)\geq \min\{v(a),v(b)\}$, with equality if $v(a)\neq v(b)$.
\end{itemize}
For a vector with entries in $\Kbar$, define its valuation to be the smallest valuation of the entries. The big O and little-o notations are used for elements in $\Kbar$ with respect to the valuation $v$. In particular, when we write $f=g+o(t^u)$ for $f,g\in\Kbar$, it means $v(f-g)>u$.

A general non-zero element in $K_r$ factors into
$ct^s(1+\sum_{i=1}^\infty a_it^\frac{i}{r}).$
Call $c$ its \emph{leading coefficient} and $ct^s$ the \emph{leading term}.

The action of $\tau$ on $K$ is given by
$$\tau(t)=\tau(\frac{1}{x})=\frac{1}{x+1}=\frac{t}{1+t}=t-t^2+t^3-\cdots,$$
which extends to $K_r$ and $\Kbar$ following
$$    \tau(t^\frac{1}{r})=t^\frac{1}{r}(1+t)^{-\frac{1}{r}}=t^\frac{1}{r}\left(1-\frac{1}{r}t+\frac{(-\frac{1}{r})(-\frac{1}{r}-1)}{2}t^2-\cdots\right).$$

We briefly describe the universal extension of $\Kbar$. For more details see \cite[Chapter 6]{van2003galois}. Denote by $\Kbar\{\hyp\}$ the algebra over $\Kbar$ generated by hypergeometric elements. Let $\tau$ act naturally on $\Kbar\{\hyp\}$. Then the polynomial ring $\Kbar\{\hyp\}[l]$ is a universal extension of $\Kbar$, equipped with a $\tau$-action following the rule
$$\tau(l)=l+t.$$
The valuation $v$ extends to $\Kbar[l]$ by setting $v(\sum_{i=0}^n a_il^i)=\min\{v(a_i):i=0,1,\ldots, n\}$.

\subsubsection{Generalized exponents}
This section will show how generalized exponents classify solutions of operators up to a factor of valuation 0. First consider the group 
\begin{equation}\label{eq:hypuptoval0}
    \{\hyp(f):f\in K_r^*\}/\{f\in K_r^*:v(f)=0\},
\end{equation}
which classifies hypergeometric elements over $K_r$ up to a factor of valuation 0. Applying the map $g\mapsto \frac{\tau(g)}{g}$, namely $\hyp(f)\mapsto f$, the group (\ref{eq:hypuptoval0})
is isomorphic to
$$\mathcal{G}_r:=K_r^*/K_{1,r},$$
where $K_{1,r}=\{\frac{\tau(f)}{f}: f\in K_r^*, v(f)=0\}$. 
\begin{lemma}\label{lemma:k1r}  $K_{1,r}=\{g\in K_r^*:v(g-1)>1\}$.
\end{lemma}
\begin{proof}
A straight-forward calculation shows that
\begin{equation}
    \frac{\tau(f)}{f}=1-v(f)t+o(t).
\end{equation}
Obviously when $v(f)=0$ the right-hand side is $1+o(t)$. 
The converse follows from the proof given in \cite[Lemma 3.2.4]{YCHA}

\end{proof}
Thus, $f,g\in K_r^*$ represent the same class in $\mathcal{G}_r$ when
$v(\frac{f}{g}-1)>1,$
which happens if and only if
$v(f-g)>v(f)+1.$
For
$$f=ct^s(1+\sum_{i=1}^\infty a_it^\frac{i}{r})\in K_r^*,\quad \text{where } r<\infty,$$
denote
$$\trunc(f)=ct^s(1+\sum_{i=1}^r a_it^\frac{i}{r})\in K_r^*.$$
Clearly $\trunc:K_\infty^*\rightarrow K_\infty^*$ is well-defined. 
Then $v(f-\trunc(f))>v(f)+1$, which means the image of $f$ in $\mathcal{G}_r$ is represented by $\trunc(f)$. 
Hence $\mathcal{G}_r$ can be identified with the following set of representatives
$$ E_r :=\trunc(K_r^*) =\{ct^s(1+\sum_{i=1}^r a_it^\frac{i}{r}):c\in\mathbb{C}^*,s\in \frac{1}{r}\mathbb{Z},a_i\in\mathbb{C}\}.$$
$$ E_{\infty} := \trunc(K_{\infty}^*) = \bigcup_{r=1}^{\infty} E_r. $$
The following is a short exact sequence of abelian groups
$$1\rightarrow  K_{1,r}\rightarrow K_r^* \xrightarrow{\trunc} E_r\rightarrow 1,$$
where the set of representatives $E_r$ is turned into a group by
$$g_1\circ g_2 := \trunc(g_1g_2),\quad g_1,g_2\in {E}_r.$$

Denote 
\begin{equation}\label{eq:kbarlh}
    \Kbar[l]_{\rm hyp} =\{\hyp(f)p : f\in \Kbar, p\in \Kbar[l]\setminus\{0\}\}\subseteq \Kbar\{\hyp\}[l].
\end{equation}
It is a multiplicative monoid. An element $h\in\Kbar[l]_{\rm hyp}$ can be written uniquely in the form
$h=\hyp(e)p,$ where $e \in E_\infty$ and $p\in \Kbar[l]$ has valuation 0. Define
$\mathrm{gen}(h)= e$,
the \emph{generalized exponent} of $h$.
Note that $\mathrm{gen}$ commutes with multiplication. 

Generalized exponents of operators are defined in \cite{YCHA}. An order $n$ operator has exactly $n$ generalized exponents in $E$, counting with multiplicity. 
The relevant property for us is:

\begin{proposition}\label{prop:genexp}
\cite[Theorem 102]{Zhou}
For a non-zero operator $L\in \Kbar[\tau]$, $e\in E_\infty$ is a generalized exponent of $L$ if and only if there exists a solution $h\in \Kbar[l]_{\rm hyp}$ with $e=\mathrm{gen}(h)$.
\end{proposition}

\begin{definition}\label{def:genexpsys} (Generalized exponents for systems).
Based on Proposition~\ref{prop:genexp}, we call $e\in E_{\infty}$ a \emph{generalized exponent} of the system $\tau(Y)=MY$ if this system has a solution of the form $\hyp(e)S$ where $S\in \Kbar[l]^n$ and $v(S)=0$.

\end{definition}
\begin{remark}[Algorithms]
For operators we can quickly compute the generalized exponents with the program GeneralizedExponents in the LREtools package in Maple 2021. For systems over $\mathbb{C}(x)\subseteq K$,
Section~\ref{sec:local} gives an algorithm for computing the {\em unramified}\, generalized exponents
(i.e. those in $\mathcal{G}_1$). 

\end{remark}
\begin{lemma}\label{rem:genexp2}
Denote $e=\trunc(f)$ for $f\in \Kbar^*$.
Due to Lemma~\ref{lemma:k1r}, $\hyp(\frac{f}{e})$ lies in $\Kbar$ and has valuation $0$. Hence if $\hyp(f)S$ is a hypergeometric solution of a system where $v(S)=0$, then $e$ is a generalized exponent.
\end{lemma}

Now we discuss the relation between hypergeometric solutions of systems over $\mathbb{C}(x)$ and their generalized exponents.
Since $\mathbb{C}(x)\subseteq K$, $\trunc(f)\in\mathcal{G}_1$ for $f\in \mathbb{C}(x)^*$. For a nonzero monic Laurent series
$$f=t^n+dt^{n+1}+\cdots \in K^*,$$
denote $\slc(f)=d$, where $\slc$ stands for \emph{second leading coefficient}. 

For a monic polynomial 
$$A=x^n+dx^{n-1}+\cdots,$$
we have
$$\trunc(A)=x^n+dx^{n-1}=t^{-\deg(A)}(1+\slc(A)t),$$
and for a non-zero rational function $c\frac{A}{B}$ where $A,B$ are monic polynomials, straight-forward computations show that
$$\trunc(c\frac{A}{B})=c t^{\deg(B)-\deg(A)}\left(1+(\slc(A)-\slc(B))t\right).$$

The following lemma justifies the definition of $\mathcal{H}$ in Step BP2 (Section~\ref{sec:algv1}).

\begin{lemma}
Consider equation~(\ref{eq:sys}).  Suppose $A,B\in C[x]$ are monic. If there is a solution in the form $\hyp(cA/B)P$ where $P\in C(x)^n$ then the system has a generalized exponent $ct^s(1+dt)\in E_1$ that satisfies
\begin{equation}\label{eq:compatibility1}
\deg(B)-\deg(A)=s,\quad \slc(A)-\slc(B)-d=v(P)\in\mathbb{Z}.  
\end{equation}
If we further require $P\in C[x]^n$ then the relations have a stronger form
    \begin{equation}\label{eq:compatibility2}
\deg(B)-\deg(A)=s,\quad -\slc(A)+\slc(B)+d=-v(P)=\deg(P)\in\mathbb{N}.
\end{equation}
We will refer to \eqref{eq:compatibility1} as the \emph{weak compatibility relations} and \eqref{eq:compatibility2} the \emph{strong compatibility relations}.
\end{lemma}
\begin{proof}
Note that
$$\hyp\left(c\frac{A}{B}\right)P   
=\hyp\left(c\frac{A}{B}(1+t)^{-v(P)}\right)\,t^{-v(P)}P,$$
where $v(t^{-v(P)}P)=0$. By Remark~\ref{rem:genexp2} the existence of such a solution implies 
$$\trunc(c\frac{A}{B})\circ \trunc((1+t)^{-v(P)})=ct^{\deg(B)-\deg(A)}(1+(\slc(A)-\slc(B)-v(P))t)$$
is a generalized exponent in $E_1$. When $P\in C[x]^n$, $-v(P)=\deg(P)\in \mathbb{N}$.
\end{proof}

\begin{example}\label{example2}
Consider the same system as in Example~\ref{example1}. There are two unramified generalized exponents:
$$t^{2}(1+2t),\quad t^{-1}(1-4t).$$
Therefore $c=1$. Eight $A/B$'s in $\mathcal{S}(M)$ from Step 1 (Example~\ref{example1}) match $t^2(1+2t)$; none matches the other generalized exponent.
\end{example}

\subsection{Step BP3} 
In this step, for each potential $\lambda = c A/B$,  compute polynomial solutions of
$\lambda \tau(P) = M P$ using
the algorithm given in \cite{AB98}.

\begin{example}[Continued from Example~\ref{example2}]\label{example3}
Our next step is to find all polynomial solutions of the system
\begin{equation}\label{eq:polysol}
   \tau(P)=(c\frac{A}{B})^{-1}MP 
\end{equation}
for each $c\frac{A}{B}$. The set of hypergeometric solutions is one dimensional, given by
$$\hyp\left(\frac{x+1}{x(x+2)(x+3)}\right)
\begin{pmatrix}
  (x+1)^3  (x+2) x^2\\ (x+2) (x+1) (-x-1)\\ (x+2)^2 (x+1)^2\\ -x^4-4 x^3-3 x^2+1\\ -x-2\\ (x+2) (x+3)
\end{pmatrix}.
$$
Algorithm Version I computes this solution four times; different $c\frac{A}{B}$'s can lead to the same solution.
If $\hyp(\lambda)P$ is a solution where $P\in C[x]^n$ and $\hyp(\lambda')=f\hyp(\lambda)$ for some $f\in C[x]$, then Version I may
rediscover $\hyp(\lambda)P$ by computing $\hyp(\lambda')fP$.
\end{example}

Computing too many candidates and duplicate solutions -- issues  (a),(b) in Section~\ref{sec:hypsolintro} -- will be addressed in Version II.

\section{Algorithm Version II} 
\subsection{Type and local types} 
\begin{definition}[Type]
  For $f_1, f_2\in C(x)^*$, say $f_1$ and $f_2$ have the same \emph{type} if $\hyp(f_1){C}(x)=\hyp(f_2){C}(x)$.
\end{definition}
Duplicates can occur in Algorithm Version I because the set $\mathcal{S}$ can contain different $\frac{A}{B}$'s with the same type. If we select one single $\frac{A}{B}$ for each type, then repeated solutions are avoided. This can be done by bounding \emph{local types}. 

Notation: For a prime polynomial $p\in C[x]$, denote
  $[p]:=\{\tau^i(p): i\in\mathbb{Z}\}$, the set of  all polynomials that are {\em shift equivalent} to $p$.

It is easy to detect if two monic irreducible polynomials are shift equivalent, because the $i$ in $f = \tau^i(g)$ can be computed by dividing $\mathrm{slc}(f)-\mathrm{slc}(g)$ by the degree,
and if $i$ is an integer, then one checks if $f$ is actually equal to $\tau^i(g)$.

For a prime polynomial $p \in C[x]$, let $v_p:C(x)\rightarrow \Z\cup\{\infty\}$ be the $p$-valuation.
\begin{definition}[Local Type]\label{def:localtype}
  For a non-zero element $a \in C(x)\setminus\{0\}$ and a prime polynomial
  $p \in C[x]$, let
  $$
    g_p(a) = \sum_{k \in \Z} v_{\tau^k(p)}(a).
  $$
  In other words, $g_p$ is the sum of valuations with respect to all prime polynomials that are shift equivalent to $p$. 
  We call $g_p(a)$ the \emph{local
    type} of $a$ at $[p]$.
\end{definition}

The following theorem shows the relation between types and local types.

\begin{theorem}[{\cite[Theorem 1]{VANHOEIJ1999109}}]\label{thm:type}
  Suppose $\lambda_1 = c_1\frac{A_1}{B_1}, \ \lambda_2 = c_2\frac{A_2}{B_2}\in C(x)^*$ where $A_1,B_1,A_2,B_2\in C[x]$ are monic and $c_1,c_2\in C$. Then $\lambda_1$ and $\lambda_2$
  have the same type if and only if
  \begin{itemize}
      \item $c_1=c_2$, and
      \item $g_p(\lambda_1)=g_p(\lambda_2)$ for any prime $p\in C[x]$.
  \end{itemize}
\end{theorem}

Let now $Y = \hyp(cA/B) P$ be once more a hypergeometric solution to
equation~(\ref{eq:sys}) where $c\in C^*$ and $A,B\in C[x]$ are monic.  
Then Theorem~\ref{thm:AB} yields $A \mid \denom(M^{-1})$ and
$B \mid \denom(M)$. The first statement implies
$$
  0 \leq g_p(A) \leq g_p(\denom(M^{-1}))
$$
for every prime polynomial $p \in C[x]$ while the second statement yields
$$
  0 \leq g_p(B) \leq g_p(\denom(M)).
$$
Now we restate Theorem~\ref{thm:AB} in terms of local types.

\begin{corollary}\label{lem:gp}
  For a hypergeometric solution $Y = \hyp(\lambda) P$ of equation~(\ref{eq:sys}), we have
  $$
    -g_p(\denom(M))
    \leq g_p(\lambda) 
    \leq g_p(\denom(M^{-1}))
  $$
  for every prime $p \in C[x]$.
\end{corollary}

If the number of combinations in the algorithm below is high, then it can be worthwhile to spend more CPU time
to obtain a sharper bound:
  \begin{equation} \label{eqDen}
    -g_p(\denom(\tau(M) M))
    \leq 2 g_p(\lambda) 
    \leq g_p(\denom( (\tau(M) M)^{-1})).
  \end{equation}
To see why this holds, note that $Y$ is also a solution of $\tau^2(Y) = \tau(MY) = \tau(M) \tau(Y) = \tau(M) M Y$, and
$\tau^2({\rm hyp}(\lambda)) = \tau(\lambda)\lambda {\rm hyp}(\lambda)$,  and $g_p( \tau(\lambda)\lambda ) = 2 g_p(\lambda)$. One can derive even sharper bounds
from denominators of even larger products ($\tau^2(M) \tau(M) M$, etc.) for diminishing returns.

\begin{remark}[Denominator bounds] \label{remDen}
Let $M_1:=M$,  $M_{i+1} := \tau^i(M) M_i$, in particular $M_2 = \tau(M) M$.
These matrices can be used not only to bound $g_p(\lambda)$, like in equation (7),
they can also be used to give a family of denominator-bounds
for rational solutions of systems, as explained in \cite{denomBound}. 
Traditionally, a denominator bound for rational solutions is computed
from ${\rm denom}(M)$ and ${\rm denom}(M^{-1})$.
But we can obtain sharper denominator bounds if we compute
${\rm denom}(M_i$) and ${\rm denom}(M_i^{-1})$ for larger $i$.
The traditional bound, $i=1$, costs the least to compute. The trade-off is that if we take say $i=2$,
we get a sharper bound, but spend more CPU time to compute it, as it involves
computing $M_2$ and $M_2^{-1}$, which could be large matrices.
\end{remark}

\subsection{Algorithm Version II}

Theorem~\ref{thm:type} and Corollary~\ref{lem:gp} lead to a more efficient algorithm for computing hypergeometric solutions. 

\begin{algorithm}[Version II]~\\
Input: $\tau(Y)=MY$, where $M\in \GL_n(C(x))$ \\
Output: hypergeometric solutions. 
\begin{itemize}[leftmargin=1cm]
    \item[BP1] 
    Factor $\denom(M)$ and $\denom(M^{-1})$. Say
    $$\denom(M)=a_1\prod_{i=1}^m p_i^{e_i}, \quad \denom(M^{-1})=a_2\prod_{i=1}^m p_i^{e_i},$$
    where $a_1,a_2\in C-\{0\}$ and $p_i$ are monic irreducible polynomials. \\
    After reordering $p_1,\ldots,p_m$ we can suppose $[p_1],[p_2],\ldots, [p_l]$ are all the shift equivalence classes. \\
    For $i=1,2,\ldots,l$, calculate $g_{p_i}(\denom(M))$ and $g_{p_i}(\denom(M^{-1}))$. \\
    $\mathcal{S}_2:=\{\prod_{i=1}^l p_i^{f_i}: -g_{p_i}(\denom(M))\leq f_i\leq g_{p_i}(\denom(M^{-1}))\}$.
    \item[BP2] 
        Use the algorithm introduced in subsection \ref{sec:d} to compute 
        $$G:=\{\text{unramified generalized exponents of the system}\}.$$
        Let $\mathcal{H}_2:=$
        $$\{c\frac{A}{B}:cx^s(1+dx^{-1})\in G,\frac{A}{B}\in \mathcal{S}_2, \deg(A)-\deg(B)=s,d-\slc(\frac{A}{B})\in \mathbb{Z}\}.$$
    \item[BP3] 
        Let $\mathrm{Sols}=\emptyset$. \\
        For each $c\frac{A}{B}\in \mathcal{H}_2$, solve the system $\tau(P)= c^{-1}\frac{B}{A}MP$ for (a basis of) \emph{rational} solutions (e.g. \cite{AB98, abramov1999}). Add $\hyp(c\frac{A}{B})P$ to $\mathrm{Sols}$ for any rational solution $P$. \\
        Return $\mathrm{Sols}$ as the output.
\end{itemize}  
\end{algorithm}

\begin{remark}\label{rem:notstandardform}
Steps BP2 and Step BP3 are slightly different from the same steps in Version I. The reason is, if $\hyp(cA/B)P$ is a hypergeometric solution in standard form, then Version II only
computes the type of $c A/B$, which determines $\hyp(cA/B)$ up to a factor in $\Ct(x)$.  As a result,
$P$ is no longer guaranteed to have polynomial entries. Hence in Step BP2 the weak compatibility relations \eqref{eq:compatibility1} are applied instead of the strong and in Step BP3 we compute rational solutions.
\end{remark}

\begin{example}\label{example4}
Let $M$ be the same matrix as in Example~\ref{example1}. 
Every irreducible factor of $\denom(M)=x(x+2)^2(x+3)$ and $\denom(M^{-1})=(x+1)^2(x+2)$ is shift equivalent to $x$, and
$$g_x(\denom(M))=4,\quad g_x(\denom(M^{-1}))=3. $$
Therefore,
$$\mathcal{S}_2(M)=\{x^i:-4\leq i\leq 3\}.$$
The cardinality of $\mathcal{S}_2(M)$ is 8, a significant improvement over Version I (Example~\ref{example1}).
In Example~\ref{example2} we calculated the unramified generalized exponents:
$$t^{2}(1+2t),\quad t^{-1}(1-4t).$$
The matching elements in $\mathcal{S}_2(M)$ are $x^{-2}, x$ and so $\mathcal{H}_2(M)=\{x^{-2},x\}$.
\end{example}

Version II has two advantages over Version I:
\begin{itemize}
    \item More efficient: there are fewer candidates for $\frac{A}{B}$.
    \item No duplicate solutions.
\end{itemize}
as well as one drawback: we now have to compute rational solutions $P$ instead of polynomial solutions (Remark~\ref{rem:notstandardform}).
But this drawback is easily overcome, the next section shows how one can ensure that $P$ will be in $C[x]^n$. 

\section{Algorithm Version III}

Algorithm Version II greatly reduced the number of $c\frac{A}{B}$'s by avoiding repeated types. The drawback in Version II is that when implemented as stated, it involves
computing {\em rational} instead of {\em polynomial} solutions for each $c\frac{A}{B}$ in Step~BP3. We address this drawback in Version III, which is very similar,
except that each $A/B$ is chosen using Theorem~\ref{thm:polsol} below, which ensures that it suffices to only compute polynomial solutions in Step~BP3.
A few notations are needed before stating the main result of this section.

 %
 Suppose $r_1,r_2\in C(x)$ are of the same type. Say $r_1 \preccurlyeq r_2$ if $\frac{\hyp(r_2)}{\hyp(r_1)}\in C[x]$. It is easy to see that $\preccurlyeq$ is a partial order (on the set of all non-zero rational functions or those that are of the same type). Theorem~\ref{thm:AB} states that every hypergeometric solution can be written as $\hyp(c\frac{A}{B})P$ where $P\in C[x]^n$. If $\hyp(c'\frac{A'}{B'})P'=\hyp(c\frac{A}{B})P$ and $c'\frac{A'}{B'}\preccurlyeq c\frac{A}{B}$, then a consequence is $P'\in C[x]^n$. Thus, to achieve the goal, we want $\frac{A}{B}$ to be as small as possible.
 
 Recall that
 $$\mathcal{S}(M)=\{\frac{A}{B}:A,B\in C[x] \text{ monic, }A\mid \denom(M^{-1}), B\mid \denom(M)\}$$
contains all potential $\frac{A}{B}$ in Algorithm Version I (Section~\ref{sec:algv1}). Consider the partition $\mathcal{S}(M)=\bigcup_i \mathcal{T}_i(M)$ of $\mathcal{S}(M)$, where each $\mathcal{T}_i(M)$ consists of all elements of a particular type. 
 \begin{theorem} \label{thm:polsol}
 Suppose $s_i\in \mathcal{T}_i(M)$ has the smallest slc in $\mathcal{T}_i(M)$. 
 Then $s_i$ is the $\preccurlyeq$-smallest element in $\mathcal{T}_i(M)$.
 \end{theorem}
   \begin{proof}
All elements in $\mathcal{T}_i(M)$ have the same local type, so the local types of $\mathcal{T}_i(M)$ is well-defined.
Suppose $\mathcal{T}_i(M)$ has a non-zero local type $e_1,e_2,\ldots, e_N$ at $[p_1],[p_2],\ldots, [p_N]$, respectively,
where $p_1,p_2,\ldots,p_N$ are mutually shift non-equivalent.
We first consider the case $N=1$.
 
 Assume $s_i$ is not the $\preccurlyeq$-smallest element in $\mathcal{T}_i(M)$.
Then there exists $r\in \mathcal{T}_i(M)$ such that $\hyp(\frac{r}{s_i})\notin C[x]$. Therefore there exists an irreducible polynomial $p\in C[x]$ such that
 $$\hyp(\frac{r}{s_i})=\frac{f}{pg},$$
 where $f,g\in C[x]$ and $\gcd(f,pg)=1$. Then $s_i\frac{p}{\tau(p)}\in \mathcal{T}_i(M)$ 
 and $$\mathrm{slc}(s_i\frac{p}{\tau(p)})=\mathrm{slc}(s_i)+\mathrm{slc}(p)-\mathrm{slc}(\tau(p))=\mathrm{slc}(s_i)-\deg(p),$$
 which contradicts to the assumption that $s_i$ has the smallest $\mathrm{slc}$ among elements in $\mathcal{T}_i(M)$.
 
The $\preccurlyeq$-smallest element is unique because the ordering is anti-symmetric.
The proof also works for $N>1$, treating each shift-equivalent class separately.
 \end{proof}
 \begin{remark} \label{degbound}
If $ct^s(1+dt)$ is the unramified generalized exponent with the largest $d$ that is compatible with $\hyp(c\frac{A}{B})P$ where $P\in C[x]^n$, then $d-\slc(A/B)$ is a degree bound for $P$ due to the relation $d=\slc(A/B)+\deg(P)$. With a degree bound the problem of finding polynomial solutions reduces to solving a system of linear equations.
\end{remark}

\begin{example}\label{example6}
Let $M$ be the same as in Example~\ref{example1}. The unramified generalized exponents are:
$$t^{2}(1+2t),\quad t^{-1}(1-4t).$$
There are eight types of elements in $\mathcal{S}(M)$, two of them matching the generalized exponents, as was stated in Example~\ref{example4}. 
The $\preccurlyeq$-smallest elements in $\mathcal{S}(M)$ of these two types are
$$
s_1  = \frac{x+1}{(x+2)^2(x+3)}\quad 
s_2  = \frac{(x+1)^2}{x+3}. $$
The set of candidates for $c\frac{A}{B}$ is $\{s_1,s_2\}$. 
The degree bounds from Remark~\ref{degbound} are $8$ for $s_1$ and $-3$ for $s_2$. So we discard $s_2$ and compute
polynomial solutions for $\tau (P)=\frac{1}{s_1}MP$, finding  
$$\begin{pmatrix}
  x^3(x+2)(x+1)^4\\ -x(x+2)(x+1)^3\\ x(x+2)^2(x+1)^3\\ -x(x+1)(x^4+4x^3+3x^2-1)\\ -x(x+2)(x+1)\\ x(x+3)(x+2)(x+1))
\end{pmatrix}.$$
\end{example}


\section{Computing the unramified generalized exponents of a system} \label{sec:local}

\subsection{Local Hypergeometric  Solutions}
The difference operator $\De = \T - 1$ is given by  $\De(f) = \T(f) - f = f(x+1)-f(x).$
The difference field $(\Ct\(x^{-1}\), \T)$ is an extension of $(\Ct(x), \T)$, so any 
system of first order linear difference equations, $\T(Y) = M Y$  with $M \in\GL_n(\Ct(x))$, can be viewed as a difference system over $\Ct\(x^{-1}\)$.
Moreover,  if it has a hypergeometric solution over $\Ct(x)$
  $$\hyp(c \frac{A}{B})\,P$$ where 
 $c \in \Ct^*$, $A, B \in \Ct[x]$ and $P  \in \Ct[x]^n$ then it has a formal solution \cite{BaCh1} of the form  
  \begin{equation}\label{FHS} Y = \Gamma(x)^{s} c^x x^d F
\end{equation}
where $s \in \Z$, $c \in \Ct^*$, $d  \in  {\Ct}$ and $F\in {\Ct}[[x^{-1}]]^n$
 with $\ord(F) =0$:
 $$F = \left(F_0 + x^{-1}F_1 + \dots \right)\, \quad F_0 \neq 0.$$\\ 
More precisely, one has the relations
$$s = \deg A - \deg B \in \Z,$$ 
$$ d - \slc(A)+\slc(B) = \deg P \in \N. $$ 

 Note that such a formal solution is in fact a hypergeometric solution over $\Ct\(x^{-1}\)$. Indeed,  it is a solution 
of the form $\gamma F$ where $\gamma = \Gamma(x)^s c^x x^d$  is hypergeometric  over $\Ct\(x^{-1}\)$, since   $\frac{\tau(\gamma)}  {\gamma} = cx^s (1+\frac{1}{x})^d = cx^s(1+ \frac{d}{x}+ \cdots)  \in {\Ct}\(x^{-1}\)$.
 For this reason  a formal solution of the form (\ref{FHS}) will be called a {\em local hypergeometric solution}  (\lhs  in short) of our system. Thus, a hypergeometric solution over $C(x)$ is also
 a local hypergeometric solution.  The set $G$ of unramified generalized exponents of the system can be identified with 
 the set of all triplets $(s, c, d) \in \Z \times \Ct^{*} \times \Ct$ involved in possible \lhs (which we still denote by $G$). 
 
 We will explain in this section how to determine this set $G$, directly from the system (i.e. without using {\em cyclic vectors}). 
 
 As in the scalar case ($L(y) =0$ with $L \in \Ct[x][\T]$), the possible values of $s$ are among the integer slopes (if any) of the {\em $\T-$Newton polygon} of $L$ \cite{BaCh1}. To each slope $s$ corresponds a Newton polynomial $P_s(\lambda) \in \Ct[\lambda]$ whose nonzero roots (if any) in $\Ct$ are candidates for the $c$'s. For each pair $(s,c)$ the corresponding $d$'s are obtained as roots in $\Ct$ of the so-called indicial equation (or Newton polynomial) associated to the slope $0$ (if any) of the {\em $\delta-$Newton polygon} of the operator $L_{s,c}$ obtained by performing in the equation $L(y)=0$ the change of variable $y \mapsto (\Gamma(x)^s c^x) \cdot y$. 
 
  Let us note that in the operator case the local data (Newton polygon and polynomials) can be easily obtained from the coefficients of the operator (see \cite{BaCh1, BaCh2}, \cite{ClvHo06}). However, in the system case the local data cannot be obtained directly from $M$ nor from its inverse. Our strategy consists in reducing the given system to a suitable gauge-equivalent one from which we can easily obtain the integer slopes of the Newton polygon and their corresponding Newton polynomials.  In the
systems case, we do not compute the Newton polygon, we only need its integer slopes and those are determined
in another way, using the properties of a {\em super-reduced form} of the system. \\

  Before explaining our method we begin by setting up some notation.   
Any difference system  
\begin{equation}\label{Stau}  
\T(Y) = M Y,  \quad \textit{where }  M \in \mathrm{GL}_{n}(\Ct(x))
\end{equation} 
can be rewritten as 
\begin{eqnarray}\label{SDelta}
\De(Y) = \widehat{M} Y,  \quad \textit{where }  \widehat{M}:= M - I_n  \in \Mat{\Ct(x)}nn.
\end{eqnarray}

Performing the {\em gauge transformation} $Y = T Z$ where $T \in\GL_n(\Ct(x))$ 
yields  the following  \emph{gauge-equivalent} system
$$
    \label{eq:Delta-Gauge}
    \De(Z) =  \wh{N } Z\quad \;\; 
    \wh{N }= T_{\De}[\wh{M}] :=  \T(T^{-1})  (\wh{M}T - \De(T)).
$$


If $Y=  \Gamma(x)^{s} c^x x^d F$ with $F_0\neq 0$ is a \lhs of the system $\T(Y) = MY$ then for all $T \in\GL_n(\Ct(x))$, $Z= TY$ is a \lhs of the equivalent system $\T(Z) = T_\delta[M]Z$.  Indeed,
$$ Z=TY = \Gamma(x)^{s} c^x x^d TF=  \Gamma(x)^{s} c^x x^{d+\nu} \wt{F}$$ where 
$\wt{F}\in {\Ct}[[x^{-1}]]^n$ with $\wt{F}_0 \neq 0$ and where $\nu \in \Z$ with $\nu \geq \ord(T)$. 
Thus, the pairs $(s,c)$  are invariants  under gauge transformations, while the $d$'s may change into $d + \nu$, for some $\nu \in \Z$.\\

We are interested  in the computation of all the triplets  $(s,c, d)  \in \Z \times {\Ct}^* \times {\Ct}$  that  appear in a \lhs of  equation~(\ref{eq:sys}).  In general, those triplets cannot be obtained directly from the matrix system $M$.  However,  we will show that
\begin{itemize}
\item[-]   first, all the possible pairs $(s,c)$ can be obtained from a super-irreducible form, in the sense of \cite{Bar89, BaBrPf08},  of  the matrix $\widehat{M}$;
\item[-]  when a pair $(s,c)$ is obtained, the corresponding indices (the $d$'s) are computed  from  a $0$-simple form (\cite{BaClEl}) of  the system $\delta Y = (x^{-s}c^{-1}M - I_n )Y$.
\end{itemize}

\subsection{Computing the pairs $(s,c)$}\label{sc-section}
For a \lhs $Y=\Gamma(x)^{s} c^x x^d F$ we will refer to $s$ as the slope, $c$ as the characteristic root and $d$ as the index of $Y$. 
The following facts are easy to prove 
\begin{enumerate}
\item The slope $s$ of any  \lhs  of  the system $\T(Y) = MY$ satisfies the following inequalities:  $$ \ord(M^{-1}) \leq s \leq -\ord(M).$$
\item  If $Y$ is a  \lhs  of slope $s$  then for any integer $p$,   $ \wt{Y}= \Gamma(x)^{p}Y$ is a \lhs with slope $\tilde{s}:=s+p$,   of  the difference system $\tau(\wt{Y}) = x^pM \wt{Y}$.   Hence, if we choose $p > - \ord(M^{-1})$  then  we can make  $\tilde{s} \geq \ord(x^{-p}M^{-1}) >0$.
Thus,  even if it means to replace $M$ by $x^pM$ with $p= - \ord(M^{-1}) +1$, one can suppose that 
$\ord(M^{-1}) \geq 1$ and in this case $\ord(M) \leq -1$ and  the slopes  of the  local hypergeometric solutions of   (\ref{Stau}) are positive.
\end{enumerate}

In the sequel, we assume that  $\ord(M^{-1}) \geq 1$  so that $\ord(M) \leq -1$ and  the slopes of  the  local hypergeometric solutions of  (\ref{Stau})  are  positive.    It follows that  the matrix $\widehat{M} = M-I_n$ has  the same valuation  as M. 
 
Put $ q:=-\ord(\wh{M}) \geq 1$ and let $1 \leq s \leq q$ be an integer.   We will look for a necessary and sufficient condition for $s$ to be a slope of some  \lhs  of (\ref{Stau})  that can be read easily from the first few coefficients in the expansion of $\wh{M}$. 
  
 For $ 1 \leq i \le n$, let $\nu^{(s)}_i = \min{(\ord(\wh{M}_{i,.}), -s)}$
 where $\wh{M}_{i,.}$ denotes the $i$th row of the matrix $\wh{M}$. 
 Put  
$$D^{(s)}:=\diag(x^{\nu^{(s)}_1}, \ldots, x^{\nu^{(s)}_n})   \;\; \hbox{  and  } \;\;   N^{(s)}:= D^{(s)}\wh{M}.$$ 

It is clear that $\ord{(N^{(s)})} \geq 0$ and  that $\ord{(D^{(s)})}\geq -s$ 
(in fact, the entries of the matrix $x^s D$ 
are polynomials in $x^{-1}$). It then makes sense to set 
$$ D^{(s)} = x^{-s} \left (D^{(s)}_0 +\Or(x^{-1})\right)  \;\; \hbox{  and  } \;\;  N^{(s)}=  N^{(s)}_0 +\Or(x^{-1}).$$

Our system  (\ref{Stau})  can be rewritten as 
\begin{eqnarray*}\label{L}
 D^{(s)}\De(Y) - N^{(s)} Y =0. 
\end{eqnarray*} 
Let us now  replace  $Y$  by the expression in (\ref{FHS}). One has  
$$ \De(Y) = \T\left(\Gamma(x)^{s} c^x x^d\right) \De(F) + \De\left(\Gamma(x)^{s} c^x x^d\right)\, F. $$
Now, we have $\De(F) = F(x+1)-F(x)= \Or(x^{-2})$ and 
\begin{multline*} \T\left(\Gamma(x)^{s} c^x x^d\right) = cx^s \left(1+ x^{-1}\right)^d \Gamma(x)^{s} c^x x^d\\=  cx^s\left(1+ dx^{-1}+\Or(x^{-2})\right) \Gamma(x)^{s} c^x x^d.
\end{multline*}
Since $s >0$, one has 
$$  \De\left( \Gamma(x)^{s} c^x x^d\right) =  
 x^s\left(c + \Or(x^{-1})\right) \Gamma(x)^{s} c^x x^d,$$ 
hence 
$$ \De\left( \Gamma(x)^{s} c^x x^dF\right) =  x^s \Gamma(x)^{s} c^x x^d \left(c F_0 +  \Or(x^{-1})\right).$$
Now 
\begin{multline*} D^{(s)}\De(Y) -N^{(s)} Y  \\=   \left(\Gamma(x)^{s} c^x x^d  \right)\left( x^{s}D^{(s)}  \left(cF_0 +  \Or(x^{-1})\right) - N^{(s)}   \left(F_0 +  \Or(x^{-1})\right)\right).
\end{multline*}
Thus 
\begin{eqnarray*}\label{AL} D^{(s)}\De(Y) - N^{(s)} Y =  
\Gamma(x)^{s} c^x x^d  \left( (c D^{(s)}_0 -N^{(s)}_0)F_0 + \Or (x^{-1}) \right).
\end{eqnarray*} 
If a  nonzero formal solution of the form (\ref{FHS})  exists then one must have 
  $$(c D^{(s)}_0 -N^{(s)}_0)F_0 =0,$$
 and therefore 
 \begin{equation}\label{carcEq}
  \det{(c D^{(s)}_0 -N^{(s)}_0)} = 0,
  \end{equation}
since $F_0 \neq 0$.
 
Thus, if the system (\ref{SDelta}) has a solution of the form (\ref{FHS})  then $c$ must annihilate  the polynomial   $E_s(\lambda):=\det{(N^{(s)}_0 - \lambda D^{(s)}_0)}$. 
However, it may happen that this polynomial  
vanishes identically in $\lambda$, in which case 
 it gives no useful  information. This motivates the following definition     

\begin{definition}($s$-simple).  System (\ref{SDelta})  is said to be $s$-{\em simple}  if the polynomial 
 $E_s(\lambda)=\det{(N^{(s)}_0- \lambda D^{(s)}_0)} \not\equiv 0$.
\end{definition}

It is important to notice that any system can be reduced to an  $s$-simple gauge-equivalent one. An algorithm to perform  such a reduction is presented in \cite{BaClEl} where it is shown that a super-irreducible system is necessarily $s$-simple  for  all $s \in \{ 1,\dots, q\}$. Hence, it is sufficient to compute a super-irreducible form of (\ref{SDelta}) to get all the  $E_s(\lambda)$ for  $ s=1, \dots, q$. This can be done using one of  the super-reduction algorithms presented in \cite{BaBrPf08, Bar89}.\\  

For $ 1\leq s  \leq q$, let $C_s$ be the set of nonzero roots of $ E_s(\lambda)$ in $\Ct$. If $C_s = \emptyset$ then there is no \lhs of slope $s$ and this value of $s$ can be discarded; otherwise the set $\{s\} \times C_s$  
gives candidates for the pairs $(s,c)$ which occur in the possible hypergeometric solutions with slope $s$.   
In summary, the following algorithm produces a set of candidates for the pairs $(s,c)$ occurring in the possible hypergeometric solutions of the input system.

\begin{algorithm}[$\textit{sc-Pairs}$]~\\  
Input: $M \in\GL_n( \Ct(x))$\\
Output: $SC$ the set of pairs $(s,c)$ 
\begin{itemize}
\item[1.] If $\ord{(M^{-1})} \leq0$ then set  $p:=-\ord{(M^{-1})} +1$ and replace the matrix $M$ by $x^pM$.
\item[2.] Let $\wh{M}:=M - I_n$ and $ q:= -\ord{(M)}$. 
\item[3.] Apply the super-reduction algorithm to the matrix $\wh{M}$. {\em This produces the list of the polynomials $E_s(\lambda)$ for  $ s=1, \dots, q$}.
\item[4.] $SC:=\emptyset$.
\item[5.]  For $s=1, \dots, q$  do 
\begin{itemize}
\item Compute the set $C_s:=\{c \in \Ct^* \; \vert \; E_s(c)=0\}$.
\item If $C_s \neq \emptyset$ then $SC:= SC \cup \{s-p\} \times C_s$.
\end{itemize}
\end{itemize}
\end{algorithm}

\subsection{Computing the exponents $d$'s associated to a pair $(s,c)$}\label{sec:d}
Let a system (\ref{Stau}) be given. We assume that we have computed the set $SC$ of the $(s,c)$-pairs belonging to \lhs of (\ref{Stau}) and that $SC$ is not empty. Let $(s,c) \in SC$. We will now explain how to determine the set of the indices $d \in \Ct$  such that the system possesses a \lhs  of the form (\ref{FHS}). \\
If $Y$ is such a solution then $\Gamma(x)^{-s}c^{-x}Y = x^dF$ is a formal solution of the system $\tau Z = M_{s,c} Z$ where $M_{s,c}:= c^{-1}x^{-s}M$ and hence $d$ (modulo $\Z$) must satisfy 
an algebraic equation which can be computed from a so called simple form of the matrix  $\wh{M}_{s,c}$.\\

Here again the system $\delta(Z) = \wh{M}_{s,c} Z$ is rewritten as 
\begin{eqnarray}\label{deltaMsc}
 xD\De(Z)-N Z =0, 
\end{eqnarray}  
with
$$D:=\diag(x^{\nu_1}, \ldots, x^{\nu_n})   \;\; \hbox{  and  } \;\;   N := xD\wh{M}_{s,c}$$ 
where $\nu_i = \min{(\ord(\wh{M}_{s,c,i})-1, 0)}$, $\wh{M}_{s,c,i}$ being the $i$th row of the matrix $\wh{M}_{s,c}$.
By construction, one  has $\ord{(N)} \geq 0$ and  $\ord{(D)}\geq 0$  so that 
$$D = \left (D_0 +\Or(x^{-1})\right)  \;\; \hbox{  and  } \;\;  N =  N_0 +\Or(x^{-1}).$$
  
Plugging the expansions of $D, N$, $Z=x^d(F_0+\Or(x^{-1}))$ and $x\De(Z) =  x^d(dF_0+ \Or(x^{-1}))$  in equation~(\ref{deltaMsc}), one gets
\begin{eqnarray*}\label{AL0}  
x^d  \left( (d D_0 -N_0)F_0 + \Or (x^{-1}) \right)=0.
\end{eqnarray*} 
which implies 
$$ (dD_0 -N_0)F_0 =0$$
and, since $F_0 \neq 0$,
\begin{eqnarray}\det\left( dD_0 -N_0 \right) =0.\end{eqnarray} 

Thus the index $d$ must be a root of the polynomial $E_0(\lambda):= \det\left( \lambda D_0 -N_0 \right)$.\\  In order to get a finite set of candidates for the $d$'s, it suffices to make sure that the system $\De(Z) = \wh{M}_{s,c} Z$ is such that $E_0(\lambda)$ is not identically zero (in which case we say that it is $0$-simple).  As mentioned before, it is always possible  to reduce a given difference system to a gauge-equivalent system which is simple \cite{BaClEl}.\\ 

The existing (and implemented) algorithms for computing simple forms \cite{BaClEl} produce a polynomial $E_0(\lambda)$ and a gauge transformation $T$ which is polynomial in $x^{-1}$ and such that its inverse $T^{-1}$ is polynomial in $x$. \\ 
Since $Tx^dF = x^{\tilde{d}}(\wt{F}_0+ \Or(x^{-1}))$, with $\wt{F}_0\neq 0$ and $E_0(\tilde{d}) =0$, it follows 
that $d \in \tilde{d}+\ord(T) + \N$. Hence we only know $d$ modulo $\N$ and this is sufficient for what we want to do.
 
\begin{algorithm}[Unramified Generalized Exponents]~\\
Input: $M \in\GL_n(\Ct(x))$\\
Output: the set  $G$ of the triplets $(s,c,d)$ involved in the \lhs of $\tau(Y) = M Y$ 
\begin{itemize}
\item[1.] $SC:=\textit{sc-Pairs}(M)$.
\item[2.] $G:=\emptyset$.
\item[3.]  For $(s,c) \in SC$  do 
\begin{itemize}
\item Let $\wh{M}_{s,c}:= c^{-1}x^{-s}M-I_n$.
\item Compute a $0$-simple form  of $\wh{M}_{s,c}$. This produces a polynomial $E_0(\lambda)$ and a gauge transformation $T$ which is polynomial in $1/x$. 
\item Compute the set $R_{sc}:=\{d  \in \Ct \; \vert \; E_0(d)=0\}$.
\item If $R_{sc} \neq \emptyset$ then $G:= G \cup \left(\{(s,c)\} \times (R_{sc}+\ord(T))\right)$.
\end{itemize}
\end{itemize}
\end{algorithm}

\section{Application: Beke-Bronstein Algorithm}
The Beke-Bronstein algorithm (\cite{Beke} with improvements in \cite{Bronstein:1994:IAF:190347.190436}) for factoring differential operators applies to polynomials and difference operators as well. In the recurrence case,
this leads to a recurrence system, and computing hypergeometric solutions. 
In the following we reformulate the Beke-Bronstein approach in the language of {\em exterior algebra}, first for the polynomial case, and then the difference case  (it works almost the same for both cases).
\subsection{Polynomial case}

Let $f\in \mathbb{Q}[y]$ be the polynomial that is to be factored. Assume $y\nmid f$ throughout this section. If $g$ is a factor of $f$ then so is $ag$ where $a\in\mathbb{Q} - \{0\}$.
To identify these factors we use $\mathbb{P}(\mathbb{Q}[x]) = (\mathbb{Q}[x]-\{0\})/(\mathbb{Q}-\{0\})$.
If $g \in \mathbb{Q}[x]-\{0\}$ then let $[g]$ denote its image in $\mathbb{P}(\mathbb{Q}[x])$. Call $[g]$ a \emph{projective factor} of $f$ if $g\mid f$. Let $\mathrm{Fact}_m(f)=\{[g]:g\mid f, \deg(g)=m\}$ 
be the set of degree-$m$ projective factors. 

Let $M=\mathbb{Q}[y]/(f)$. Then $M$ is a $\mathbb{Q}[y]$-module.
For $g\mid f$ with $\deg(g)=m$, let $\mu(g)=g\wedge yg\wedge \cdots \wedge y^{n-m-1}g\in \bigwedge^{n-m}M$ where $n$ is the degree of $f$. Suppose $g=\sum_{i=0}^m b_iy^i$. 
\begin{lemma}\label{lemma:mug}
The coordinate of $\mu(g)$ with respect to $y^i\wedge y^{m+1}\wedge y^{m+2}\wedge \cdots \wedge y^{n-1}$ is $b_i(b_m)^{n-m-1}$. 
\end{lemma}
\begin{proof}
It is a straight-forward calculation.
\end{proof}
If $y\nmid f$ then $\mu(g)\neq 0$. For $a\in \mathbb{Q}^*$, we have $\mu(ag)=a^{n-m}\mu(g)\neq 0$. Thus 
$$\mu([g]):=\mathbb{Q}\mu(g)\subset \bigwedge^{n-m}M$$
is well-defined; $\mu([g]) \in \mathbb{P}(\bigwedge^{n-m}M)$ does not depend on the choice of representatives. 
Lemma~\ref{lemma:mug} shows how to recover $[g]$ from $\mu([g])$. This gives a map $\eta:\mathbb{P}(\bigwedge^{n-m}M)\rightarrow \mathbb{P}(\mathbb{Q}[y])$
such that $\eta\circ \mu$ is the identity on $\mathrm{Fact}_m(f)$.

\begin{proposition}\label{thm:bronstein} $\mu(\mathrm{Fact}_m(f)) \subseteq E_m(f)$ where
$$E_m(f):=\{\text{1-dimensional subspaces of }\bigwedge^{n-m}M \text{ that are }\mathbb{Q}[y]\text{-modules}\}.$$
\end{proposition} \vspace{3pt}


Now $\mathbb{Q}v$ is in $E_m(f)$ if and only if $v$ is an eigenvector
of the action of $y$ on $\bigwedge^{n-m}M$.
Lemma~\ref{lemma:mug} converts such $v$ to a polynomial $g$.
Then $g \in \mathrm{Fact}_m(f)$ (i.e. $g | f$) if and only
if the so-called Pl\"ucker relations hold, which in this
case are equivalent to $(y^i g) \wedge v = 0$ for $i < n-m$.  (We may omit $i=0$ since $g \wedge v = 0$ will hold by construction.) Indeed,
if $g | f$, then the corresponding $v = \mu(g)$ is the wedge product of $y^i g$ for $i < n-m$,
and hence $(y^i g) \wedge v = 0$. 

\begin{example}
Suppose $f=y^4+2y^3+3y^2+2y+2$. The goal is to illustrate Beke-Bronstein by computing
factors of $f$ of degree 2.  
In $M=\mathbb{Q}[y]/(f)$ the relation
$$y^4=-2y^3-3y^2-2y-2$$
holds. From that we compute the action of $y$ on a basis of $\bigwedge^2 M$:
\begin{equation*}
\begin{split}
{1\wedge y^3} & \mapsto y\wedge y^4= y\wedge (-2y^3-3y^2-2y-2) \\
{y\wedge y^3} & \mapsto y^2\wedge y^4= y^2\wedge (-2y^3-3y^2-2y-2) \\
{y^2\wedge y^3} & \mapsto y^3\wedge y^4= y^3\wedge (-2y^3-3y^2-2y-2) \\
1\wedge y & \mapsto y\wedge y^2 \\
1\wedge y^2& \mapsto y\wedge y^3 \\
y\wedge y^2& \mapsto y^2\wedge y^3 \\
\end{split}
\end{equation*}
which can be represented by the matrix
$$
\begin{pmatrix}
    0 & -2 & 0 & 2 & 0 & -3  \\
   0 & 0 & -2 & 0 & 2& 2 \\
2    & 2&3 &0 &0 &0 \\
   0 & 0 & 0 &0 & 0 & 1 \\
   0 &1 &0 &0 &0 & 0\\
   0 &0 &1 &0 & 0& 0
\end{pmatrix}.
$$
The eigenvectors are
$\left( -1 \ 0 \ -1 \ -1 \ 0 \ 1 \right)$ and $\left( 1 \ 1 \ 1/2 \ 2 \ 2 \ 1 \right)$.
    The Pl\"ucker relation(s) for $n=4$, $m=2$ are well known and are given by a single polynomial 
    $$X_{0,1}X_{2,3}-X_{0,2}X_{1,3}+X_{0,3}X_{1,2},$$
    where $X_{i,j}$ is the coordinate corresponding to the basis vector $y^i\wedge y^j$. Both eigenvectors satisfy this Plücker relation.
    To obtain the corresponding factors of $f$, we take the $y^i \wedge y^3$ coefficients of the above eigenvectors. We sorted the
    basis of $\bigwedge^2 M$ in such a way that these coefficients are the first three entries of the eigenvectors. This
   leads to the factors 
    $$-1-y^2, \quad 1+y+\frac{1}{2}y^2.$$
\end{example}

\subsection{Difference case}\label{sec:bronsteindifference}

Let $D=\mathbb{Q}(x)[\tau]$, let $L,R \in D$, and let
$[R] = \mathbb{Q}(x) R$. We call $[R]$
a \emph{projective factor} of $L$ if $R$ is a right-hand factor of $L$.
As before, $$\mathrm{Fact}_m(L):=\{[R]:R \text{ is an order-}m \text{ right-hand factor of }L\}.$$ Denote $M=D/DL$. Define
$$\mu:\mathrm{Fact}_m(f)\rightarrow \bigwedge^{n-m}M,$$
$$[R]\mapsto \mathbb{Q}(x) \mu(R)$$
where $\mu(R) = R\wedge \tau R\wedge \cdots \wedge \tau^{n-m-1}R$.
Similar to Theorem~\ref{thm:bronstein} we have
\begin{proposition} $\mu(\mathrm{Fact}_m(L)) \subseteq E_m(L)$ where
$$E_m(L):=\{1\text{-dimensional }\mathbb{Q}(x)\text{-subspaces of }\bigwedge^{n-m}M\text{ that are }D\text{-modules}\}.$$
\end{proposition} \vspace{3pt}

As in the polynomial case, finding $\mathrm{Fact}_m(L)$
reduces to finding 1-dimensional submodules of $\bigwedge^{n-m}M$. Here this
means finding hypergeometric solutions.

\begin{example}
Let $L =$
\begin{multline*}
    (x^3+5x-6) \tau^4+(x^4+2x^3+6x^2-3x-18)\tau^3-(x^4+5x^3+14x^2+28x-12)\tau^2 \\ -(x^5+3x^4+8x^3+3x^2-21x-18)\tau+{3x^2(x^2+3x+8)}.
\end{multline*}
To find order-$2$ factors of $L$, consider the $D$-module $\bigwedge^2 D/DL$. Like before, computing the action of $\tau$ on $\bigwedge^2 D/DL$ with respect to the basis
$$e_0=1\wedge \tau^3, e_1=\tau\wedge \tau^3,e_2=\tau^2\wedge \tau^3,e_3=1\wedge \tau,e_4=1\wedge \tau^2,e_5=\tau\wedge\tau^2$$
gives a matrix very similar to the one in the previous example.
It has two hypergeometric solutions, just like the previous example had two eigenvectors, and again both satisfy the Plücker relation(s). From the first three entries one finds:
$$\tau^2-x,\quad \tau^2-x\tau-3.$$
\end{example}


\section{Implementations and experiments} \label{SectionNew}

We have two factoring implementations. The first implementation~\cite{HypSolsSystems} computes the hypergeometric solutions with Algorithm Version III.
This implementation currently uses an old implementation for super-reduction from 1997.
It should reduce matrices over $C((x^{-1}))$, however, at the moment
it computes over $C(x)$, which can cause expression swell.  The implementation is efficient for large matrices over $C[x,1/x]$. This suggests
that a version that uses truncated power series (increasing precision when needed) should be equally efficient.

The second factoring implementation, written by van Hoeij, is available in the LREtools package in Maple.
It computes generalized exponents \cite[Section 3.2]{YCHA}, not of the matrix, but of the operator $L$, and then constructs the generalized exponents of the matrix from them.
The implementation for this is very efficient; truncated power series are used to avoid unnecessary expression swell.

Next, the factoring implementation constructs, for each individual entry of the solution vector, a bound for the denominator, a bound for the content of the numerator, and a degree bound.
For more on this, see~\cite[Chapter 6]{Zhou} which describes how degree bounds for right-factors $R$ of $L$ are computed from the generalized exponents of $L$.

Source code for generalized exponents, degree bounds and factoring is available at \url{http://www.math.fsu.edu/~hoeij/algorithms/RFactors}.
The code is also in Maple's LREtools package. The help page contains further information about generalized exponents.
Degree bounds~\cite[Chapter 6]{Zhou} computed from generalized exponents are also used in {\tt LREtools[MinimalRecurrence]},
which computes the provably minimal order recurrence (for a sequence given by a recurrence and initial terms).


We successfully tested the second implementation for factors of order $m=6$ of operators of order $n=12$.
This involves computing hypergeometric solutions of a system of order $N$ by $N$,
where $N ={\rm binomial}(n,m) = 924$. With a cyclic vector approach this would not be possible because the cyclic vector operator
would be astronomically large. In fact, order $n=8$ is already problematic.

As an example, we took $L$ to be the LCLM (Least Common Left Multiple) of $L_1 := \tau^4+\tau^3-x \tau^2 + 5\tau-x$ and $L_2 := \tau^4+x \tau^3-x^2 \tau + 2x + 1$.  While LCLM's tend to increase the degree in $x$,
the operator $L$ is still of modest size, the degree in $x$ is only 10.
By construction, $L$ has two very small right-factors of order $m=4$.
Computing factors of $L$ of orders $m=1,2,3,4$ with
our implementation {\tt LREtools[RightFactors]} in Maple produces $\emptyset, \emptyset, \emptyset, \{ L_1, L_2\}$ in a fraction of a second ($t_{\rm LRE}$ in the table below).
To do this with the cyclic vector method \cite{AbramovPetkovsekRyabenko2015} we have to compute the cyclic vector operator (time $t_{\rm cv}$) for the $D$-modules $\bigwedge^m D/DL$,
and then compute its hypergeometric solutions (time $t_{\rm cv}^{\rm sol}$ ).

The problem is the growth of these operators. The order is binomial(8,$m$) which is reasonable. More problematic growth is in the degree in $x$ and the overall
size (measured by the Maple {\tt length} command).
\begin{center}
\begin{tabular}{  | c | c | c |   c | c | c | c | } \hline
$m$  & order & degree & size & $t_{\rm cv}$ & $t_{\rm cv}^{\rm sol}$ & $t_{\rm LRE}$ \\ \hline
1   &   8  & 10 & 980 & 0 & 0.06 s & 0.06 s\\ \hline
2    &  28  &  268 & 1572804 & 16.5 s & 1.57 s & 0.08 s\\ \hline
3  &  56 & 1382 & 99128859 & 4 h & 15 h  & 0.06 s   \\ \hline
4  & 70 & ? & ?  & $>2$ d& ? & 0.12 s  \\ \hline
\end{tabular}
\end{center}
As the size grows, it takes a very long time to compute these cyclic vector operators. The cyclic vector operator for $m=3$ took 4 hours to compute and Maple reported an additional 15
hours to compute its hypergeometric solutions. For $m=4$ we have not yet determined the cyclic vector operator, it is still running after more than 2 days, but even
if we find it, it is likely that computing its hypergeometric solutions will be time consuming.

In contrast, {\tt LREtools[RightFactors](L,4)} only takes 0.12 seconds to find the factors of order $m=4$.
%
The reason for this dramatic difference in CPU time is the fact that, although the dimension of the exterior power matrix grows as ${\rm binomial}(8,m)$,
the degrees and {\tt length} of the entries in the matrix do not grow with $m$.
In contrast, the coefficients of the cyclic vector operators for these matrices grow dramatically with $m$, both in degree and in {\tt length}.

In 2004, the second author implemented the cyclic vector approach for both the differential and the difference case.
The difference case was shared with a few colleagues while the differential case was added to Maple's {\tt DEtools[DFactor]} in 2004.
Only order $n=4$ was implemented because expression swell made the cyclic vector approach inefficient for higher order
(even for $n=4$ the newer  {\tt LREtools[RightFactors]} is much faster).

We experimented with the $p$-curvature approach from \cite{ClvHo06} to reduce the number of cases for $c A/B$. After
experiments we settled on using just one prime, $p=2$, in our implementation \cite{HypSolsSystems}, and none in {\tt LREtools[RightFactors]}.
The question remains how to devise a good strategy for the number of primes (if any) to be used.  Equation~(\ref{eqDen}) and Remark~\ref{remDen}
can also reduce the number of cases for $c A/B$. We have not yet used this because in experiments
the number of cases was often not high after comparing ${\rm deg}(A)-{\rm deg}(B)$ and ${\rm slc}(A/B)$ with $s,d$ in step BP2.
But it is certainly possible that there could be applications where the number of cases remains high,
in which case the above mentioned improvements should be added to the implementation.

\bibliographystyle{plain}

\begin{thebibliography}{10}

\bibitem{HypSolsSystems}
Implementation.
\newblock \url{http://www.math.fsu.edu/~hoeij/files/HypSols}.

\bibitem{abramov1999}
Sergei Abramov.
\newblock Eg-eliminations.
\newblock {\em Journal of Difference Equations and Applications}, 5:393--433,
  01 1999.

\bibitem{AB98}
Sergei~A. Abramov and Moulay~A. Barkatou.
\newblock Rational solutions of first order linear difference systems.
\newblock In {\em Proceedings of ISSAC'98}, Rostock, 1998.

\bibitem{AbramovPetkovsekRyabenko2015}
Sergei~A.\ Abramov, Marko Petkov{\v s}ek, and Anna~A. Ryabenko.
\newblock {\em Hypergeometric Solutions of First-Order Linear Difference
  Systems with Rational-Function Coefficients}, pages 1--14.
\newblock Springer International Publishing, 2015.

\bibitem{Bar89}
Moulay~A. Barkatou.
\newblock {\em Contribution \`a l'\'etude des \'equations diff\'erentielles et
  aux diff\'erences dans le champ complexe}.
\newblock PhD thesis, l'{I}nstitut {N}ational {P}olytechnique de Grenoble,
  France, June 1989.

\bibitem{BaBrPf08}
Moulay~A. Barkatou, Gary Broughton, and Eckhard Pfl\"{u}gel.
\newblock Regular systems of linear functional equations and applications.
\newblock In {\em ISSAC '08: Proceedings of the twenty-first international
  symposium on Symbolic and algebraic computation}, pages 15--22, New York, NY,
  USA, 2008. ACM.

\bibitem{BaCh1}
Moulay~A.\ Barkatou and G.~Chen.
\newblock Computing the exponential part of a formal fundamental matrix
  solution of a linear difference system.
\newblock {\em J. Difference Eq. Appl.}, 1996.

\bibitem{BaCh2}
Moulay~A.\ Barkatou and G.~Chen.
\newblock Some formal invariants of linear difference systems.
\newblock {\em Pub. IRMA Lille}, 44(2), 1997.

\bibitem{BaClEl}
Moulay~A. Barkatou, Thomas Cluzeau, and Ali El{-}Hajj.
\newblock Simple forms and rational solutions of pseudo-linear systems.
\newblock In James~H. Davenport, Dongming Wang, Manuel Kauers, and Russell~J.
  Bradford, editors, {\em Proceedings of the 2019 on International Symposium on
  Symbolic and Algebraic Computation, {ISSAC} 2019, Beijing, China, July 15-18,
  2019}, pages 26--33. {ACM}, 2019.

\bibitem{Beke}
Beke E.
\newblock Die {I}rreduzibilität der homogenen linearen
  {D}ifferentialgleichungen.
\newblock {\em Math. Ann.}, 45:278--294, 1894.


\bibitem{bronsteinprivate}
Manuel Bronstein.
\newblock Personal communication.

\bibitem{Bronstein:1994:IAF:190347.190436}
Manuel Bronstein.
\newblock An improved algorithm for factoring linear ordinary differential
  operators.
\newblock In {\em Proceedings of the International Symposium on Symbolic and
  Algebraic Computation}, ISSAC '94, pages 336--340, New York, NY, USA, 1994.
  ACM.

\bibitem{YCHA}
Yongjae Cha.
\newblock {\em Closed Form Solutions of Linear Difference Equations}.
\newblock PhD thesis, Florida State University, 2010.

\bibitem{ClvHo06}
Thomas Cluzeau and Mark van Hoeij.
\newblock Computing hypergeometric solutions of linear recurrence equations.
\newblock {\em Appl. Algebra Eng. Commun. Comput.}, pages 83--115, 2006.



\bibitem{VANHOEIJ1999109}
Mark van Hoeij.
\newblock Finite singularities and hypergeometric solutions of linear
  recurrence equations.
\newblock {\em Journal of Pure and Applied Algebra}, 139(1):109 -- 131, 1999.

\bibitem{denomBound}
Mark van Hoeij, Moulay Barkatou, and Johannes Middeke.
\newblock A Family of Denominator Bounds for First Order Linear Recurrence Systems.
\newblock {\em Preprint} arXiv:2007.02926, 2020.


\bibitem{Petkovsek1992}
Marko Petkov{\v s}ek.
\newblock {Hypergeometric solutions of linear recurrences with polynomial coefficients},
\newblock  {\em J. Symb. Comput.}
1992 {\bf 14}, 243--264.



\bibitem{van2003galois}
M.~van~der Put and M.F. Singer.
\newblock {\em Galois Theory of Linear Difference Equations}.
\newblock Springer-Verlag, Heidelberg 1997.



\bibitem{walkeralgcurves}
Robert~J. Walker.
\newblock {\em Algebraic Curves}.
\newblock Springer-Verlag New York, 1978.

\bibitem{Zhou}
Yi~Zhou.
\newblock {\em Algorithms for Factoring Linear Recurrence Relations}.
\newblock PhD thesis, Florida State University, USA, April 2022.

\end{thebibliography}

\end{document}